\documentclass{svmult}
 
\usepackage{enumerate}
\usepackage{epsf,epsfig,a4wide}  
 \usepackage[applemac]{inputenc}
\usepackage[dvipsnames]{xcolor}

 \usepackage{lmodern}
\usepackage{enumitem}
\usepackage{mathdots}

\usepackage{url}
\usepackage{amsmath,amssymb}

\usepackage{newtxtext,newtxmath}

\usepackage{helvet}         
\usepackage{courier}        
\usepackage{makeidx}         
\usepackage{graphicx}        
\usepackage{multicol}        
\usepackage[bottom]{footmisc}
\usepackage{amsmath,amssymb,amsfonts}

\usepackage{soul}
\makeindex  

 \newtheorem{Proposition}{Proposition}[section]

\newtheorem{fact}{Fact}[section]

\newcommand{\be}{\begin{equation}}
\newcommand{\ee}{\end{equation}}

\newcommand{\Id}{\textrm{\rm Id}}

\newcommand{\ddd}{\mathrm{d}}

\newcommand{\trace}{\operatorname{tr}}
\newcommand{\const}{\mathrm{const}}

\newcommand{\weg}[1]{}

\begin{document}
\title*{Thoughts about potentials with finite-band spectrum and finite-dimensional reductions of integrable  systems }
\author{ 
Andrey Yu. Konyaev and 
Vladimir S. Matveev} 
\institute{ 
\and Andrey Yu. Konyaev \at  Faculty of Mechanics and Mathematics, Moscow State University, 119992, Moscow, Russia
 \email maodzund@yandex.ru
  \and Vladimir S. Matveev \at 
Institut f\"ur Mathematik, Friedrich Schiller Universit\"at Jena,
07737 Jena, Germany \email  vadimir.matveev@uni-jena.de}  
\titlerunning{Potentials with finite-band spectrum and  BKM systems }
\authorrunning{  A. Konyaev, V. Matveev}
\maketitle

\abstract{
 We repeat, using methods developed for BKM systems, the famous results of S. Novikov \cite{novikov74}, J. Moser \cite{moser80,moser81}, and A. Veselov \cite{veselov80} that relate Schr\"odinger-Hill operators with finite-band spectra, solutions of the Neumann system, and certain solutions of the KdV equations. Our general motivation is to determine whether it is possible to apply inverse scattering methods to BKM systems, and in the conclusion, we indicate initial observations in this direction.}

 \section{Introduction}

J\"urgen Moser, in his very  influential and widely  cited works \cite{M83,moser80,moser81} attracted attention to relation of the following three a priory independent objects of interest  of mathematical physics, namely the    KdV equation, which is an $\infty$-dimensional integrable system,  finite-band spectrum of  
Schr\"odinger-Hill operators, and classical 
finitely-dimensional integrable systems corresponding to the geodesics of ellipsoid and the Neumann system. We give necessary definitions later, in  \S \ref{sec:neumann}.  The works \cite{M83, moser80,moser81}, though containing new interesting results, are of rather survey nature. Certain connections discussed by Moser were published  by other authors. In particular, the relation between the orbits  of the Neumann system and the finite-gap solutions of the KdV equation is discussed in \cite{alber79, veselov80}. The  relation 
between geodesics on ellipsoid and solutions of the  Neumann system is understood at least  in \cite{knorrer82}. In \cite{novikov74}, it was shown that stationary solutions of KdV  give  potentials such that the corresponding  Schr\"odinger-Hill operators  have  finitely many connected components  of spectrum. In   \cite{novikov74} it was shown that any quasi-periodic potential with finitely many connected components  of spectrum of the corresponding  Schr\"odinger-Hill operator   is generated by this way.   
See also \cite{DMB76, DKN2001}.

J. Moser, and also many other mathematics of that time,  were  very excited about the relations and possibly viewed them as a wondeful coincidence, a kind of  mathematical magic. The goal of the paper is to give a short proof of the relations. We will use new understanding coming from the recently found connection  between finitely-dimensional integrable systems and geodesically equivalent metrics, see e.g. \cite{MT97,BKMseparation} and  also the very recent results of \cite{BKMreduction} on finite-dimensional reductions of  BKM-systems\footnote{BKM systems  is a family of multicomponent integrable PDE-systems    introduced in \cite{nijapp4}. In the present note we mostly concentrate on KdV equation, which is a BKM system. In the Conclusion and Outlook section \ref{sec:con}, we comment on possible generalizations for other BKM systems.}

Our note is organised as follows. First we introduce and discuss Benenti systems. All systems which we consider: the Neumann system, geodesics on the ellipsoid and the finite-dimensional reductions of the KdV equation can be described by certain Benenti systems, which  we explicitly write. Next, we    
prove a relatively simple  Lemma \ref{lem:2} which explains  when two Benenti system share the same solutions. The relation of the KdV system to the Neumann system and of the geodesic flow of the metric geodesically equivalent to the metric of ellipsoid  follows  directly from this Lemma.     
Lemma \ref{lem:3}  specifies the conditions that the parameters of the solutions of the KdV equation corresponding to the Neumann system must satisfy.

 Next, we consider the Sturm-Liouville-Schr\"odinger equation whose potential comes from a solution of the finite-dimensional reduction of the KdV equation. We show  in Lemma \ref{lem:p2} that one can explicitly write the solutions of this equation using solutions  of the  Benenti systems used  in  the finite-dimensional reduction of the KdV equation. Next,  solutions of finite-dimensional reduction of the KdV equation that correspond to solutions of the Neumann systems can be studied using method of separation of variables.  Employing this method, we  show that the spectrum contains finitely many bands. 
\section{Preliminaries}

\subsection{Benenti systems}

By 
 {\it Benenti systems}\footnote{The terminology was suggested in \cite{BM03,IMM}}
 we understand   a class of integrable finite-dimensional Hamiltonian systems on the cotangent bundle $T^*M$.  The Hamiltonian is the sum of the kinetic energy $\tfrac{1}{2}g^{ij}p_ip_j$, where $g$ is a metric of any signature, and a potential energy $V:M\to \mathbb{R}$. The metric $g$ and the potential energy $U$ satisfy certain differential-geometric and nondegeneracy conditions which we introduce and discuss now. We will describe these systems, and give a local classification of them in this section. We will see that, under nondegeneracy assumptions which are fulfilled in  the   cases  we   consisder in the present paper, Benenti systems are given by two functions, $f$ and $U$, of one variable.

We start with a pair $(\textrm{metric $g$}, \ \textrm{$g$-seladjoint (1,1)-tensor $L$})$ satisfying the following condition: 
\begin{equation} \label{eq:M0}
    L_{ij, k} = \lambda_{j} g_{ij}+ \lambda_{i} g_{jk}.
\end{equation}
In the equation above we use $g$ for index  manipulation  and  denote by comma the covariant derivative with respect to the Levi-Civita connection of  $g$. The 1-form $\lambda_i$ staying on the right hand side is necessarily $\lambda_i= \tfrac{1}{2} d\trace_g(L)$.

The equations \eqref{eq:M0} appeared many times independently in different branches of mathematics. In particular,  in the context of integrable systems,      solutions $L$  of this equation, for a given $g$,  are  called  in \cite{ Crampin03,CST} {\it  special conformal Killing tensors}. See also  \cite{Bbook, MB02}.

In differential geometry, \eqref{eq:M0}  appeared   at least in \cite{Sinjukov79},  in the context of geodesically equivalent metrics, see also \cite{BM03,Matveev07}.    The equation \eqref{eq:M0} is  now  called {\it geodesic or projective compatibility} equation, see e.g. \cite{BMR21, BM11}. Recall that two metrics $g$ and $\bar g$ on the same manifold are {\it geodesically (or, which is the same, projectively) equivalent}, if every $g$-geodesic, after an appropriate reparameterisation, is a $\bar g$-geodesic. The relation of \eqref{eq:M0} to geodesic equivalence is as follows: if 
$L$ is nondegenerate and satisfies \eqref{eq:M0}, then the metric  $\bar g_{ij}$  whose matrix is given by 
\begin{equation} \label{eq:barg}
    \bar g= \tfrac{1}{\textrm{det}(L)} g  L^{-1} 
\end{equation}
is geodesically equivalent to $g$. Moreover, if $\bar g$ is geodesically equivalent to $g$, then $L$ reconstructed by \eqref{eq:barg} satisfies \eqref{eq:M0}.

Given $L^i_j$ satisfying \eqref{eq:M0}, one can construct a family of commutative integrals for the geodesic flow of $g$. Namely, consider the following family of functions on the tangent bundle $T^*M$ of the manifold $M$  of dimension $N$,   which is a polynomial in $\lambda $ of degree $N-1$:{\small
\begin{equation} \label{eq:MI}
   p\in T^*M \ \mapsto \  I_\lambda(p)=\tfrac{1}{2} \det(\lambda \Id-L) g^*((L^*-\lambda\Id )^{-1} p, p)= I_{0}(p)\lambda^{N-1} + I_{1}(p)\lambda^{N-2}+ \cdots + I_{N-1}(p) .
\end{equation}}
In the above formula,
$g^*$ is the induced bilinear form on the cotangent bundle, and $L^*$ is the induced operator on the cotangent bundle. In the matrix notation, the matrix of $g^*$ is the inverse of the matrix $g$ and the matrix of $L^*$ is the transposed of the matrix $L$. In the index notation, the middle part  of \eqref{eq:MI} reads 
$$
 \tfrac{1}{2}\det(\lambda \Id- L)  g^{ij} \left((L-\lambda \Id)^{-1}\right)^{s}_ip_s p_j.
$$

\begin{theorem}[\cite{BM03,CST, Crampin03, MT97, Topalov,TM2003}] \label{thm:MT}
The functions $I_i$ commute with respect to  the  standard  Poisson structure on  $T^*M$. 
\end{theorem}

Note that the function  $I_0$ is  the Hamiltonian $\tfrac{1}{2}g^*(p, p)$ of the geodesic flow.

The number of  functionally independent integrals $I_i$ is the degree of the minimal polynomial of $L$ at least at one point, see \cite[Lemma 5.6 and Corollary 5.7]{Matveev07} or \cite[Theorem 2 and Proposition 3]{Topalov}.  
Later we assume that $L$ is differentially nondegenerate\footnote{see \cite[Definition 2.10]{nij1}} almost everywhere. Then,  by \cite[\S 4.2]{nij1}, at almost every point,  the degree of the minimal polynomial of $L$  equals the dimension of the manifold\footnote{(1,1)-tensors such that the degree of the minimal polynomial of $L$  equals the dimension of the manifold are called gl-regular, see \cite{nij3}} which implies that    the geodesic flow of $g$ is  Liouville integrable. As it will be clear later, near almost every point, the integrable system can be effectively  analyzed  by the  method of separation of variables.   

Let us give two examples of geodesically compatible pairs $(g,L)$ such that $L$ is differentially-nondegenerate.
The first is classical and is essentially due to  T. Levi-Civita \cite{Levi-Civita}. 
The  metric  $g_{ij}$ and the tensor  $L_j^i$ are given by  the formulas  
\begin{eqnarray} \label{eq:LC}
   g^f_{LC}& = & \sum_{i=1}^N \left(\tfrac{dq_i^2}{f(q_i)}\prod_{s=1, s\ne i}^{N} (q_i-q_s) \right), \\ 
  L_{\textrm{diag}}   &= & \textrm{diag}(q_1,q_2,\dots,q_N) . \label{eq:Ldiag}
\end{eqnarray}
Above $f$ is an arbitrary function of one variable such that it is 
never zero. 
Note that in a more standard way to write the formula \eqref{eq:LC}, the functions $f(q_i)$ in the denominators  may be different functions, so the $i$th term in the sum \eqref{eq:LC} reads  $$\left(\tfrac{dq_i^2}{f_i(q_i)}\prod_{s=1, s\ne i}^{N} (q_i-q_s)\right).$$ This way is clearly equivalent to the one in \eqref{eq:LC}, at least in a sufficiently small neighborhood, 
since the formula gives a non-degenerate metric  only if $q_i\ne q_j$ for $i\ne j$, and under this assumption we can define 
the functions $f(q_i)$  as $f_i(q_i)$ restricted to the range of the coordinate $q_i$.

\begin{remark}
   In the examples interesting for the present paper, the function  $f$ is a real-analytic function.  This is not a special property of our examples, but a rather general phenomenon. Indeed, the metric in the examples is real analytic.
By
\cite{KM16},   $L$ satisfying \eqref{eq:M0} is real analytic as well. If $L$ is differentially nondegenerate, which is the case for the Benenti systems  coming from finite-dimensional reductions of BKM system,
by \cite[Theorem 6.2]{nij1} $f$ must be  a globally defined  real-analytic function. Moreover, if 
 $g$ is Riemannian and $L$ is gl-regular at least in one point, which is the case for the Benenti systems coming from the Neumann systems and from the geodesic flow of the ellipsoid, $f$ again must be a globally defined real-analytic function by \cite{M06}. 
\end{remark}

For the further use, note that the metric $\bar g^f_{LC}$ geodesically equivalent to $g^f_{LC}$ constructed by \eqref{eq:barg} reads 
\begin{eqnarray} \label{eq:bargLC}
   \bar g^f_{LC}& = & \sum_{i=1}^N \left(\tfrac{(q^i)^{N-3}\ddd q_i^2}{f(q_i)}\prod_{s=1, s\ne i}^{N} \frac{(q_i-q_s)}{q_iq_s} \right).  
\end{eqnarray}

Let us also note that, if we call $g_{LC}^1$ the metric \eqref{eq:LC} with $f(t)\equiv 1$, then one can obtain 
the metric \eqref{eq:LC}, i.e., the metric $g^f_{LC}$,  by the formula 
\begin{equation} \label{eq:gf0} 
g^f_{LC}=  g^1_{LC}f(L)^{-1}, 
\end{equation}
see e.g. \cite[\S 1.3]{BM11} for the definition and discussion  of analytic functions of $(1,1)$-tensors.  Observe also that the metric 
$g_1$  can be invariantly  characterised within $g$ satisfying  \eqref{eq:M0} with respect to a fixed differentially nondegenerate $L$. Indeed,  the coordinates such that  $(g, L)$ are given by
(\ref{eq:LC}, \ref{eq:Ldiag}) are geometrically distinguished as  they are eigenvalues  of $L$. 

The second example   was studied  in particular in \cite[Proposition 6.1]{nij1}, see also \cite{appnij2}.  In the context of integrable systems, it can be found e.g.  in \cite[\S V]{BM06}. 
We first start with the  contravariant metric $g$ and the (1,1)-tensor $L^i_j$  given by  
\begin{equation} \label{eq:comp}
g_1^{ij}= 
\begin{pmatrix}    
0 & \cdots & 0 & 0 & \!\!\!\! 1\\ 
0 &  \cdots & 0  &\! 1 & w_1\\ 
\vdots & \iddots &\iddots & \iddots  &w_{{2}} \\ 
0&\!\!1&w_{{1}}& \iddots &   \vdots    \\
1&w_{{1}}&w_{{2}}  &\cdots & \!\!\!  w_{{N-1}}
  \end{pmatrix} \, \  \  L_{\textrm{comp}} = \begin{pmatrix}
-w_1 & 1 & 0 & \cdots  & 0 \\
-w_2 & 0 & 1 & \ddots &   \vdots\\
\vdots & \vdots & \ddots & \ddots & 0 \\
-w_{N-1} &  0 & \dots & 0 & 1\\
-w_N & 0  & \dots & 0 & 0 
\end{pmatrix}.
\end{equation}

The metric $g_1$ in this example is flat and is of  splitted  signature. 
By \cite[Theorem 6.2.]{nij1}, in the real-analytic category, any (contravariant)  metric $g$ which is  geodesically compatible with $L_{\textrm{comp}}$  from \eqref{eq:comp}   is given by 
\begin{equation} \label{eq:gf}
   g_f =  f(L_{\textrm{comp}})g_1.   
\end{equation}

Note that the visual difference between the formulas \eqref{eq:gf0}   and  \eqref{eq:gf}  is artificial since   \eqref{eq:gf}  is for a  contravariant metric and \eqref{eq:gf0} is for a covariant metric. In fact, the formulas  \eqref{eq:gf}  and \eqref{eq:gf0}  are equivalent. 
Moreover, the  pair $(g_f, L_{\textrm{comp}})$ given by \eqref{eq:comp}, in a neighborhood of any point such that  $L$ has $N$ real different eigenvalues,  
it isomorphic by a coordinate change  to the pair   $(g^f_{LC}, L_{\textrm{diag}})$.  The coordinate change $w(q)$
 transforming $(g_f, L_{\textrm{comp}})$ to  $(g^f_{LC}, L_{\textrm{diag}})$ is given by the following relation: 
\begin{equation} \label{eq:trans}
    \det(\lambda \Id - L_{\textrm{diag}(q)}) = \det(\lambda \Id - L_{\textrm{comp}(w)}) 
\end{equation}
Note that as $L_{\textrm{comp}}$ is in the companion form, \begin{equation} \label{eq:uq}\det(\lambda \Id - L_{\textrm{comp}})= \lambda^N+ w_1 \lambda^{N-1}+\cdots+w_N, \end{equation}
so \eqref{eq:trans} gives  an explicit formula $w(q)$ in terms of symmetric polynomials of $q$-variables: 
$$
 w_1= -q_1-q_2-\cdots-q_N\ , \ \  w_2= q_1q_2 + q_1q_3+ \cdots+ q_{N-1}q_N \ , \ \dots \ , \ w_n=(-1)^N \det(L_{\textrm{diag}}).
$$
This observation holds, with necessary natural amendments, in the case when some  different eigenvalues of $L$ are complex-valued; of course in this case the coordinates in which $(g, L)$ has the form (\ref{eq:LC}, \ref{eq:Ldiag}) may also be complex-valued. See \cite{nij1, BM15} for a discussion of complex eigenvalues of $L$ satisfying \eqref{eq:M0} and of Nijenhuis\footnote{(1,1)-tensor is called {\it Nijenhuis operator}, if its Nijenhuis torsion vanished. BKM systems were constructed in \cite{nijapp4} within the Nijenhuis geometry project initiated in \cite{BMMT, nij1}} operators with complex eigenvalues.

Let us now discuss the integrals. It appears that the metric $g^f_{LC}$, in coordinates $q_i$, can be obtained by the so-called St\"ackel construction, see e.g. \cite{BKMreduction}. 
Indeed, take the St\"ackel matrix 
\begin{equation}\label{eq:st1}
S_{ij}= \begin{pmatrix} (q_1)^{N-1} & (q_1)^{N-2}& \cdots & 1 \\ 
                               (q_2)^{N-1} & (q_2)^{N-2}& \cdots & 1 \\
                                   \vdots & & & \vdots \\ 
                                   (q_N)^{N-1} & (q_N)^{N-2}& \cdots & 1\end{pmatrix}
\end{equation}
and construct the functions $I_0,\dots, I_{N-1}$  on $\mathbb{R}^{2N}(q,p)$, quadratic in $p$-variables,  by the   formula 
\begin{equation}\label{eq:st2}
\begin{pmatrix} (q_1)^{N-1} & (q_1)^{N-2}& \cdots & 1 \\ 
                               (q_2)^{N-1} & (q_2)^{N-2}& \cdots & 1 \\
                                   \vdots & & & \vdots \\ 
                                   (q_N)^{N-1} & (q_1)^{N-2}& \cdots & 1\end{pmatrix} \begin{pmatrix} I_0 \\ 
                               I_1 \\
                                   \vdots   \\ 
                                   I_N\end{pmatrix} = 
\begin{pmatrix} -\tfrac{1}{2}f(q_1)(p_1)^2 \\ 
                            -\tfrac{1}{2} f(q_2)   (p_2)^2 \\
                                   \vdots   \\ 
                                  -\tfrac{1}{2} f(q_N) (p_N)^2\end{pmatrix}                                  . 
\end{equation}
It is known, see e.g. \cite[Fact 1.2]{BKMreduction}, that the  functions Poisson commute with respect to the standard Poisson structure. Moreover, the functions $I_0,\dots, I_{N-1}$ staying in \eqref{eq:st2} coincide with those  obtained via \eqref{eq:MI}. In particular, 
the function 
$I_0$ is the Hamiltonian of the geodesic flow of the metric $g^f_{LC} $ given by \eqref{eq:LC}.

It is also known how to ``introduce'' the potential energy in the formula \eqref{eq:st2}. Locally, in coordinates $q_1,\dots, q_N$ used in \eqref{eq:st2}, the 
freedom is the choice of $N$ functions of one variable. In our context,   these functions  
are essentially the same  function, so we proceed with  the  construction under  this assumption\footnote{Actually, if the $(g, L)$ are given by \eqref{eq:comp}, in the analytic category and near the point 
$(w_1=0,\dots,w_N=0)$,  these    $N$ functions  can be glued in one function; this again follows from \cite[Theorem
6.2]{nij1}}.
In order to do it, we slightly modify \eqref{eq:st2} to obtain 
\begin{equation}\label{eq:st3}
\begin{pmatrix} (q_1)^{N-1} & (q_1)^{N-2}& \cdots & 1 \\ 
                               (q_2)^{N-1} & (q_2)^{N-2}& \cdots & 1 \\
                                   \vdots & & & \vdots \\ 
                                   (q_N)^{N-1} & (q_N)^{N-2}& \cdots & 1\end{pmatrix} \begin{pmatrix} \tilde I_0 \\ 
                               \tilde I_1 \\
                                   \vdots   \\ 
                                  \tilde  I_{N-1}\end{pmatrix} = 
\begin{pmatrix}-\tfrac{1}{2}f(q_1) (p_1)^2 + U(q_1)\\ 
                     -\tfrac{1}{2} f(q_2)         (p_2)^2 + U(q_2)\\
                                   \vdots   \\ 
                           -\tfrac{1}{2} f(q_N)       (p_N)^2+ U(q_N)\end{pmatrix}                            
.\end{equation}
The functions $\tilde I_i$ given by \eqref{eq:st3} Poisson commute. Each of them is the sum of the   quadratic in momenta part  which is $I_i$ given by \eqref{eq:st2}, and  a function which depends on the position $(q_1,\dots, q_N)$ only. In particular, the function $\tilde I_0$ is 
 the sum of the kinetic energy 
coming from the metric $g^f_{LC}$ and a potential energy.

We will also need the formula for obtained system in the companion coordinates, in which $g$ and $L$ are given by  \eqref{eq:comp}.  
In order to do it, for a real analytic function $U$ of one variable  and for a pair $(g, L)$ satisfying \eqref{eq:M0} and such that $L$ is gl-regular,  consider 
the functions 
$U_1,U_2,\dots, U_N$ defined by the following relation: 
\begin{equation} \label{eq:st4}
U(L)=U_0L^{N-1}+U_1L^{N-1}+\dots+U_{N-1}\Id. 
\end{equation}
The left  hand side is an analytic function of $L$, it is well-defined (1,1)-tensor provided $U$ is defined at the eigenvalues of $L$, 
see e.g. the discussion in \cite{nij1,BM11}. The right hand side is a polynomial of order $N-1$ in $L$.  Since $L$ is gl-regular,   \eqref{eq:st4} viewed as a system of linear equations on functions $U_0,...,U_{N-1}$  determines the functions $U_0,\dots,U_{N-1}$ on the manifold. 

\begin{lemma}  \label{lem:potential} Suppose $(g, L)$ are geodesically compatible and $L$ is gl-regular. 

Then, for any real-analytic function $U$, 
the functions $\tilde I_i:= I_i+ U_i$, where $I_i$ is as in \eqref{eq:MI} and $U_i $ are   from \eqref{eq:st4}, pairwise commute. 
\end{lemma}

Lemma \ref{lem:potential} was obtained within and is a natural component  of the  theory  of  conservation laws and symmetries of  Nijenhuis operators  developed in \cite{nij4,appnij5, nij3}, where its more general versions  are proved.  The version staying above is easy to prove directly: the formula \eqref{eq:st4} is invariant with respect to coordinates changes, so one can prove it in any coordinate system. In the ``diagonal'' coordinate system, in which   $g$ and $L$ have the form
(\ref{eq:LC},\ref{eq:Ldiag}), the formula \eqref{eq:st4} is equivalent to \eqref{eq:st3}.

We see that a Benenti system  with potential is given by the following data: we can choose  $L$  and  functions $f, U$ of one variables. Our nondegeneracy condition on the pair is that  $L$ is differentially nondegenerate almost everywhere.  This  condition  implies that the pair $(g, L)$ is  given by \eqref{eq:comp} almost everywhere, and by (\ref{eq:LC}, \ref{eq:Ldiag}) almost everywhere  provided we allow  some of the coordinates to be complex-valued.

\begin{remark} \label{rem:1}
If $L=L_{old}$ is a solution of \eqref{eq:M0}, then for any constants $\textrm{const}_1$ and $\textrm{const}_2$ the (1,1)-tensor $L_{new}= \textrm{const}_1 L_{old} + \textrm{const}_2 \, \Id $ is also a solution.  The change from  $L_{old}$ to $L_{new} $ corresponds to the following change of $f$ and $U$: 
\begin{equation}
\label{eq:fnew} f_{old}(t)= (\textrm{const}_1)^{-1-N}f_{new}(\textrm{const}_1\,  t  + \textrm{const}_2) \ , \ \ U_{old}= U_{new}(\textrm{const}_1 \, t  + \textrm{const}_2).  \end{equation}
This change    does not affect    the vector space generated by the integrals $\tilde I_0,\dots, \tilde I_{N-1}$. \end{remark}

\subsection{Neumann system,  geodesic flow  on ellipsoid and stationary solutions of the KdV equation     as  Benenti systems} \label{sec:neumann}
\subsubsection{Neumann system, geodesic flow of the metric of the ellipsoid, and the geodesic flow of the metric geodesically equivalent to ellipsoid as Benenti systems}
{\it Neumann system} is the Hamiltonian system  describing the movement of a particle on  the sphere in a quadratic potential. That is, we consider  the standard sphere $S^N\subset 
\mathbb{R}^{N + 1}(X_1,\dots, X_{N+1})$ 
with the standard metric which we denote $g_s$, and the function $V:S^N\to \mathbb{R}  $
which is the restriction of the function $\sum_{i=1}^{N+1}(X_i)^2a_i$ to the sphere. We assume that all $a_i$ are different, without loss of generality  
$a_1 < a_2 < \dots < a_{N + 1}$. Note that  adding  a constant to the potential energy does not affect the equations of motion and adds the constant to all $a_i$, so we can assume without loss of generality that $0<a_1$.
The Hamiltonian of the Neumann system is the sum of the kinetic energy coming from $g_s$ and the potential energy  $V$.

We will also consider the   ellipsoid  in  $ \mathbb{R}^{N + 1}$     defined by the equation 
\begin{equation}\label{eq:ellipsoid}
\sum_{i=1}^{N+1} \frac{X_i^2}{a_i}=1.     
\end{equation} 
We assume  that all $a_i$  are positive and different\footnote{The assumption that all $a_i$ are different is indeed important for us; the assumption that they are positive can be omitted and  instead of ellipsoid we can consider a noncompact quadric and  make sense of the case when one $a_i=0$}  and order  them by $0<a_1<\cdots <a_{N+1}$.

   The restriction of the standard metric to the ellipsoid will be denoted by $g_e$.   The corresponding geodesic flow is the Hamiltonian system whose Hamiltonian is the kinetic energy corresponding to $g_e$.

\begin{fact}  \label{fact:1} The metric $g_e$ of the ellipsoid and the standard  metric  $g_s$ of the  sphere admits  $L$ satisfying \eqref{eq:comp} 
    which is differentially nondegenerate at almost every point. 
\end{fact}

For the standard sphere, Fact  \ref{fact:1} follows from a  result of  E.  Beltrami\footnote{who constructed examples of metrics  geodesically equivalent to the metric of the sphere}, see e.g. \cite{Matveev06}.  For the ellipsoid, Fact \ref{fact:1} was independently and almost simultaneously  obtained in \cite{MT97,Tabachnikov}, see also \cite{MT01}. 

Let us give a short proof of Fact \ref{fact:1}, whose ingredients will be useful  later. We start with the ellipsoid 
and consider  the ellipsoidal   coordinates $q_1, ..., q_{N}$   related to the standard coordinates $X_1,...,X_{N+1}$ in the  quadrant $\{ (X_1,\dots, X_N) \in  \mathbb{R}^{N+1}  \mid \varepsilon_1 X_1>0,\dots, \varepsilon_{N+1} X_{N+1}>0\} $, where $\varepsilon_i\in \{-1,1\}$,  
by the formula  
\begin{equation}\label{eq:elliptic}
X_i=\varepsilon_i \sqrt{a_i\frac{\prod_{j=1}^N(a_i-q_j)}{\prod_{j=1, j\ne i}^{N+1 }(a_i-a_j)}}. 
\end{equation}

In the ellipsoidal  coordinates,   
  the metric $g_e$ of the ellipsoid  
has the  form  \eqref{eq:LC} with 
\begin{equation}f(t)=-\tfrac{4}{t}{\prod_{j=1}^{N+1}(t- a_j) }. \label{eq:fellipsoid}\end{equation}  Then,  $L= \textrm{diag}(q_1,\dots, q_N)$ is a solution of \eqref{eq:comp}. 

The calculations have  shown the local existence of  such $L$ in   a  neighborhood of a point where elliptic coordinates are defined. Clearly,  the ellipsoid is simply-connected and   the metric of ellipsoid is real analytic. Next, recall that \eqref{eq:comp}  viewed as a system of PDEs   on  $L$ is  of finite type\footnote{It closes  after 2 prolongations, see e.g. \cite{Eastwood}}. Then,    the local existence implies the global existence.

See \cite[\S 7]{MT97}
and \cite[\S 7]{MT01} for the formulas for the geodesically equivalent metric and the (1,1)-tensor $L$ in the coordinates $X_1,\dots, X_{N+1}$.
    
Later, we will use also the formula, in the ellipsoidal  coordinates,  for the metric geodesically equivalent to the metric of the ellipsoid. In view of 
\eqref{eq:barg}, it is given by
\begin{eqnarray}
    \label{eq:barge1}
     \bar g_e &= &  -\tfrac{1}{4} \sum_{i=1}^N \left( {q_i^{N-2}\left(\ddd q_i\right)^2} \prod_{s=1, s\ne i}^{N} \tfrac{(q_i-q_s)}{q_iq_s} \prod_{s=1}^{N+1}\tfrac{1}{(a_s-q_i) } \right) \\ &=& 
                 -\tfrac{1}{4} \prod_{s=1}^{N+1}a_i  \sum_{i=1}^N {q_i}\left(  \left(\ddd \tfrac{1}{ q_i}\right)^2 \prod_{s=1, s\ne i}^{N} \left(\tfrac{1}{q_s}- \tfrac{1}{q_i}\right)\prod_{s=1}^{N+1}\tfrac{1}{\tfrac{1}{q_i}-\tfrac{1}{a_s}} \right)      \label{eq:barge2} \\ &=& 
                 -\tfrac{1}{4}  \prod_{s=1}^{N+1}\tfrac{1}{\bar a_i}  \sum_{i=1}^N\left(\ddd y_i\right)^2\left(\tfrac{1}{y_i}\prod_{s=1}^{N+1}\tfrac{1}{y_i-\bar a_s}\right)  \left( \prod_{s=1, s\ne i}^{N} \left(y_s- y_i\right) \right),       \label{eq:barge3} 
\end{eqnarray}
where the new coordinates $y_i$ are given by    $y_{i}=\tfrac{1}{q_i}$ and $\bar a_i= \tfrac{1}{a_i}.$ We see   that the metric \eqref{eq:barge3} is the metric \eqref{eq:LC} corresponding  to the functions 
\begin{equation}\label{eq:fbarge}
    f(t)=  -{4} t \prod_{s=1}^{N+1} (t-\bar a_s) \prod_{s=1}^{N+1}\bar a_s.  
\end{equation}

Similarly, for the standard sphere consider the sphero-ellipsoidal coordinates\footnote{Also called sphero-conical coordinates. }
related to standard coordinates $X_1,...,X_{N+1}$ in $\mathbb{R}^{N+1}$
 by 
\begin{equation}\label{eq:spheroelliptic}
X_i=\varepsilon_i \sqrt{\frac{\prod_{j=1}^{N}(a_i-q_j)}{\prod_{j=1, j\ne i}^{N+1}(a_i-a_j)}}. 
\end{equation}

In these coordinates, the standard metric of the   sphere has the form  \eqref{eq:LC} with  \begin{equation} \label{eq:fsphere} f(t)=-4{\prod_{j=1}^{N+1}(t-a_j) }. \end{equation}
We again see that $L= \textrm{diag}(q_1,\dots,q_N)$    satisfies \eqref{eq:M0}.

Now, by direct calculations we see that the potential energy  of the Neumann system  has the form 
\begin{equation}  \label{eq:potsep}
  \sum_{i=1}^{N+1} X_i^2a_i= \tfrac{1}{2}\left( \sum_{i=1}^{N+1} a_i - \sum_{i=1}^{N}q_i\right). 
\end{equation}
By direct calculations we see that it corresponds to the function\footnote{The second term $t^{N-1} \sum_{i=1}^{N+1}a_i $ in  \eqref{eq:Usphere}
can be ignored as it corresponds to the addition of the  constant   $\sum_{i=1}^{N+1} a_i$ in \eqref{eq:potsep}  and does not change the equations of motion}  \begin{equation} U(t)= -t^{N} +t^{N-1} \sum_{i=1}^{N+1}a_i .\label{eq:Usphere} \end{equation}

Let us summarise the content of this section. We recalled that two  finitely-dimensional systems of interest, the geodesic flow of the metric geodesically equivalent to the metric of the ellipsoid, and the Neumann system, belong to the class of Benenti systems. They correspond to the following 
functions $f, U$ used in  \eqref{eq:st3}: The geodesic flow of the metric geodesically equivalent to the metric of the ellipsoid  corresponds to  $f$ given by \eqref{eq:barge1}  and to $U= 0$, and the Neumann system corresponds to $f$ given by \eqref{eq:fsphere} and to $U$ given by \eqref{eq:Usphere}.
\begin{remark}\label{rem:2}
 Let us observe from \eqref{eq:potsep} that  the potential energy corresponding to the Neumann system is equal to $\textrm{const}- \textrm{trace}(L)$. In   coordinates such that $L$ is given by \eqref{eq:comp}, trace of $L$ clearly  equals  $-w_1$.
\end{remark}
\begin{remark}\label{rem:2a}
 The range of the coordinates $(q_1,\dots, q_N)$, both in the case of ellipsoidal coordinates and sphero-ellipsoidal coordinates, is given by $a_1<q_1<a_2<q_2<\cdots<q_N<a_{N+1}$.
\end{remark}

\subsubsection{KdV equation  as BKM  system with  n=1 and its  finite-dimensional reductions. }
\label{sec:kdv}

The KdV equation is the following partial differential equation on the unknown function $u$ of two variables $x$ and $t$:
\begin{equation}
   \label{eq:kdv} 
   u_t= -\tfrac{1}{2}u_{xxx} + \tfrac{3}{2} u u_{x}. 
\end{equation}

There are different almost equivalent forms of the  equation. For example, by re-scaling of $x$ and $t$ one can change the coefficients on the right hand side of the equation, and by multiplying $u$ and $t$ by 
$-1$ one can change the signs on the right hand side. In particular, the 
  version of the KdV equation staying in  \cite{M83} and in  Wikipedia is 
\begin{equation}
\label{eq:kdv1} 
   u_t= -u_{xxx} + 6 u u_{x}. 
\end{equation}

Note also that  the transformation  $u_{new}= u_{old} + \textrm{const}$ makes from \eqref{eq:kdv} the equation 
\begin{equation}
   \label{eq:kdv2} 
   u_t= -\tfrac{1}{2}u_{xxx} + \tfrac{3}{2} u u_{x} + \textrm{const}\,\tfrac{3}{2} u_{x} . 
\end{equation}
Similarly the transformation $u_{new}(x,t)= u_{old}(x  +\textrm{const} \, t, t) $ makes from  \eqref{eq:kdv} the equation 
\begin{equation}
   \label{eq:kdv3} 
   u_t= -\tfrac{1}{2}u_{xxx} + \tfrac{3}{2} u u_{x} -  \textrm{const}  u_{x} . 
\end{equation}
Combining \eqref{eq:kdv2} with \eqref{eq:kdv3}, we see that by the appropriate choice of constants  the transformations compensate one another.

One may view  the KdV equation, and actually all BKM systems, as dynamical systems on the  space of real-analytic curves. Indeed, for a real-analytic function  $x\mapsto u(x)$, the solution $u(x,t)$
of the KdV equation  such that $u(x,0)= u(x)$  exists and is unique by the Kovalevskaya Theorem and  can be viewed as a family of functions  $x\mapsto u(x,t)$  depending on the parameter $t$.

By {\it a finite-dimensional reduction}  of the KdV equation one understands 
a finite-dimensional family of functions  $x \mapsto u(x)$ such that the family is invariant with respect to the dynamical system above\footnote{In other words, if $u(x)$ belongs to this family, then for  the solution $u(x,t)$ of the KdV equation such that $u(x,0)= u(x)$, the curves $x\mapsto u(x,t)$ lies in this family for any sufficiently small $t$}.   This notion naturally  generalises to the BKM systems.

The  finite-dimensional reduction of the KdV equation related to the topic of the paper was suggested in \cite{novikov74}, where S. Novikov  considered the so called {\it stationary solutions}; we will call it {\it stationary reduction}. The corresponding family of functions  is   defined as follows.

It is known, that the KdV equation admits symmetries, that are partial differential equations of the form 
\begin{equation}
   \label{eq:symmetry} 
   u_t= B[u],  
\end{equation}
where $B[u]$ is a differential polynomial, that is, a polynomial in $  u_x:=\tfrac{\partial u}{\partial x}$, $u_{xx}:=\tfrac{\partial u}{\partial x} $ etc.  satisfying the following property:  
 The system of equations 
\begin{equation}
       u_t= -\tfrac{1}{2}u_{xxx} + \tfrac{3}{2} u u_{x} \  \ \textrm{and} \ \        u_\tau= B[u]  , 
\end{equation}
 viewed as a system of equations on the unknown function $u(x,t, \tau)$ is compatible.

It is known and easy to see that  the KdV equation has a  trivial symmetry of  differential degree $1$
\begin{equation}
    \label{eq:trivial} u_t=u_x. 
\end{equation}
The right hand side of this symmetry  will be denoted by $B_0$, so the trivial symmetry reads  $u_t= B_0[u]$.

The KdV equation  itself is also a symmetry of differential degree 3, we denote its right hand side by $B_1$. The next symmetry  has differential degree $5$, we denote the corresponding right hand side by $B_2$ and so on, so the $N$th nontrivial symmetry, whose right hand side is denoted by $B_N$,  has differential degree $2N+1$.

It is known that for any $i,j$ the  differential symmetry  constructed by $B_i$  is symmetry for the differential symmetry constructed by $B_j$.

As a finite-dimensional family of curves, S. Novikov  \cite{novikov74} considered the family of solutions of  the 
ordinary differential equation on $u$ given by 
\begin{equation} \label{eq:sym} B_N[u]+ \lambda_1 B_{1}[u]+...+\lambda_{N-1} B_{N-1}[u] =0   \end{equation}  with constant coefficients $\lambda_1,\dots, \lambda_{N-1}$.  Because for every $i,j$ the  differential symmetry  constructed by $B_j$ is a symmetry of that 
constructed by  $B_i$, the family of curves  is invariant with respect to the flow of all differential symmetries. Indeed,  if we take  a solution $u(x,t)$ of the PDE $u_t= B_i[u]$ such that the curve $x\mapsto u(x,0)$ is a solution of $B_N[u]=0$, 
then for any $t$ the curve $x\mapsto u(x,t)$  is a solution of $B_N[u]=0$, so the family of curves is invariant with respect to the flow of \eqref{eq:kdv}. For the fixed choice of $\lambda_1,\dots ,\lambda_{N-1}$, the dimension of this family of the curves is $2N+1$, as $2N+1$  initial data  
determine  the solution of ODE of degree $2N+1$ on one unknown function $u(x)$. 
If we identify curves of the form $u(x)$ and $u(x +\textrm{const})$, the dimension of the family is $2N$, it coincides with the dimension of the cotangent space to the $N$-dimensional manifold.  The solutions of the KdV equations such that for any $t$ the curve $x\mapsto u(x,0)$ solves the ODE \eqref{eq:sym} are called {\it  stationary solutions}\footnote{There exist different nonequivalent definitions of stationary solutions of KdV in the literature. In particular, in   \cite{BSM23} a much more restrictive stationarity condition $B_N[u]=0$ is considered}.

As explained in \cite[Example 3.1]{nijapp4}, KdV equation in the form \eqref{eq:kdv} is the BKM IV system corresponding to a special choice of parameters.   In \cite{BKMreduction}, a finite-dimensional reduction of BKM systems is  studied.  Though visually the finite-dimensional reduction procedure in \cite{BKMreduction} is  different from that of in e.g. \cite{novikov74}, 
it appears that for the KdV systems the finite-dimensional reduction constructed in \cite{BKMreduction} coincides with the one considered in \cite{novikov74}.

The following results of  \cite{BKMreduction} are relevant to the present paper. The parameters of the finite-dimensional reduction from  \cite{BKMreduction}
are the number $N$ (which has nothing to do with the natural number  $N$ from the construction of BKM systems) and a monic polynomial \begin{equation} \label{eq:monic} 
C(\mu)=\mu^{2N+1} + 0 \mu^{2N} + c_2\mu^{2N-1} + \cdots+ c_{2N+1}   \end{equation}
of degree $2N+1$, whose second highest coefficient is zero.    The reduction is constructed as follows: consider the following polynomial in $\mu$ whose coefficients  depend  on $x$:
\begin{equation} \label{eq:11a} w(x;\mu) = \mu^N + \mu^{N-1} w_1(x)+\mu^{N-2} w_2(x)+\cdots+ w_N(x)  
\end{equation}
Next, consider the system of ODEs on the functions $w_1,\dots, w_n$ given by  
\begin{equation} \label{eq:base1}
 {C(\mu)}=  {m_0}\left( w_{xx} w - \tfrac{1}{2}w_x^2\right)+(\mu-\tfrac{1}{2}w_1) w^2.   
\end{equation}
We observe that both sides of   \eqref{eq:base1} are  polynomial in $\mu$ of degree $2N+1$  such that  the free terms  and linear terms  on the right and left side coincide, so \eqref{eq:base1}   is a system of $2N$   ordinary differential equations of the second order  on $n$ unknown functions $w_1,...,w_N$.   The system depends on $2N$ parameters $c_2, c_3,\dots, c_{2N+1}$ and on the parameter $m_0$. 
\footnote{
Though a generic system of $2N$  ordinary differential equation on $N$ functions does not have a solutions, the system  \eqref{eq:base1} can be solved for any values of the parameters $c_2, c_3,\dots, c_{2N+1}$ 
and the solution depends on the choice of initial values $w_1(x_0),...,w_N(x_0)$, so for fixed initial constants $c_2,\dots, c_{2N+1}$ the family of functions  is, locally,  $N$-dimensional. In fact, solutions of the system are trajectories of a Lagrangian system whose Lagrangian is the sum of kinetic energy  coming from a  flat metric and the  potential energy of a special form.   
}

Next, for any coefficients  $c_2, c_3,\dots, c_{2N+1}$ consider the following family of 1-dimensional curves given by  $x\mapsto u(x)= 2 w_1(x)$, where $w_1(x)$  comes from  a  solution of \eqref{eq:base1}. 

\begin{theorem}[\cite{BKMreduction}] \label{thm:1}
    For every coefficients  $c_2, c_3,\dots, c_{2N+1}$, the family of the curves above  is invariant with respect to the flow of $KdV$.
\end{theorem}

This gives us a finite-dimensional reduction of the KdV equation. Actually, the result of \cite{BKMreduction} is valid for all BKM systems, not only for the KdV. For the KdV systems   the reduction is essentially the stationary reduction  to the one constructed in   \cite{novikov74}. Namely, for every  function   $x\mapsto u(x)$ from the  family corresponding  to  the finite-dimensional reduction  from 
\cite{novikov74},  there exist constants $c_2,\dots, c_{2N+1}$ such that the function  lies in the family described above.

The next result of \cite{BKMreduction} which will be used in the present paper, and   which also explains why  the visually overdetermined system of equations  \eqref{eq:base1} has solutions, is as follows: 

We consider the Benenti system  constructed by $(g, L_{\textrm{comp}})$ from \eqref{eq:comp},
so the function $f(t)=1$,   and by the function\footnote{The function $U(t)$ depends on the  choice of parameters $c_2,\dots, c_{N+1}$}
\begin{equation} \label{form:U(t)}
    U(t)=  \tfrac{1}{m_0}\left(t^{2N +1}+   \sum_{i=2}^{N+1} c_{i} t^{2N-i+1}\right) =   \tfrac{1}{m_0}\left( C( t)- c_{N+2}t^{N-1}- c_{N+3}t^{N-2}-\cdots- c_{2N+1}\right). 
\end{equation}

\begin{theorem}[\cite{BKMreduction}] \label{thm:2}
    For every  parameters  $c_2,\dots, c_{N+1}$, the trajectories  of this Benenti systems viewed as functions   $x\mapsto w(x)$   on $\mathbb{R}^N$ are solutions of \eqref{eq:base1} corresponding to the polynomial $C$. Moreover,  any solution of  \eqref{eq:base1} is a trajectory of this Benenti system. 
\end{theorem}

Let us give additional explanations on the relations between parameters  $c_2,\dots, c_{N+1}$ of the potential energy in the Benenti system and the parameters $c_2, \dots, c_{2N+1}$
of \eqref{eq:base1}. The parameters of the Benenti system are the first $N$ parameters of \eqref{eq:base1}. The additional $N$  parameters  $ c_{N+2},\dots, c_{2N+1}$ can be thought as the  values of the integrals $\tilde I_0,...,\tilde I_{N-1}$ corresponding to the Benenti system, so a solution of \eqref{eq:base1} corresponding to the parameters $c_2,\dots, c_{2N+1}$ corresponds to a trajectory of the Benenti system corresponding to the parameters $c_2,\dots, c_{N+1}$ such that  the values of the integrals $\tilde I_0,\dots,\tilde I_{N-1}$ on this trajectory are $\tilde I_0=-  c_{N+2}, \tilde I_1= -  c_{N+1},\dots,  \tilde I_{N-1}= -  c_{2N+1}$.

Combining Theorems \ref{thm:1} and \ref{thm:2}, we see that,
for a fixed $t$,  
all  stationary solutions of KdV, 
viewed as the curves $x\mapsto u(x)$,  can be described as follows: we take  the Benenti system above, and for any    its solution 
$(w_1(x),...,w_N(x))$ set $u(x)= 2 w_1(x)$. The curves  $u(x)$ constructed by this method are stationary solutions of KdV and for any  stationary solution $u(x,t)$ of KdV the curve $x\mapsto u(x,t)$ can be obtained by this procedure.

\begin{remark} \label{rem:4}
Though the unknown function $u(x,t)$ in the  KdV equation depends on two variables, $x$ and $t$, the dependence on the variable $t$ was not used in   Theorems \ref{thm:1}, \ref{thm:2} and plays no role in our  paper. The dependence on $t$, for general BKM systems, was  understood as a part of  the finite-dimensional-reduction approach of \cite{BKMreduction}, and corresponds to the flow of one of the integrals. The KdV case was understood ways before. 

Though we do not use  the dependence of $t$ at all in the present paper,  let us  note that  for a solution $u(x,t)$ the   curves $x\mapsto u(x,t_1)$ and $x\mapsto u(x,t_2)$  will correspond to the same spectrum of the corresponding Schr\"odinger operator.  
\end{remark}
\section{When solutions of one Benenti system are solutions of another}

\begin{lemma} \label{lem:2}
  Consider two  families  of Benenti systems sharing the same $L$: 
 one constructed by $U_1(t)$  and $f_1(t)$, another by $U_2$ and $f_2$. Suppose for certain values of the integrals $\tilde I^1_0= H_0^1,...,\tilde I^1_{N-1}=H_{N-1}^1 $  of the first system and for certain values of the integrals  
$\tilde I^2_0= H_0^2,...,\tilde I^2_{N-1}=H_{N-1}^2 $  of the second system the products 
we have \begin{equation} \label{eq:condition} 
(-U_1(t) +  H_0^1 t^{N-1}+ \cdots + H_{N-2}^1 t  +  H_{N-1}^1  )f_1(t)=(-U_2(t) +  H_0^2 t^{N-1}+ \cdots + H_{N-2}^2 t  +  H_{N-1}^2  )f_2(t).
\end{equation} 
Then, every solution of the  first Benenti  system, viewed as a curve on $\mathbb{R}^N$, corresponding to the values of the integrals $H_0^1,\dots, H_{N-1}^1 $, is a solution of the second Benenti  system corresponding to the values of the integrals $H_0^2,\dots, H_{N-1}^2.$  
\end{lemma}

\begin{proof} Without loss of generality we may assume that $(g,L)$ are given by (\ref{eq:LC}, \ref{eq:Ldiag}). 
First we view  solutions of  both system as  curves $(q(\tau),p(\tau) )$ on $\mathbb{R}^{2N} = T^*\mathbb{R}^N$. In view of \eqref{eq:st3}, 
we have for any $i=1,\dots, N$  
\begin{equation}\label{eq:legendre-1}
\tfrac{1}{2}f (q_i)p_i^2 = -U (q_i)+ H_0  q_i^{N-1}+ \cdots + H_{N-2}  q_i  +  H_{N-1} .  
\end{equation}
Applying   Legendre transformation,  we see that \eqref{eq:legendre-1} is equivalent to the system of ODEs 
\begin{equation}\label{eq:legendre}
\tfrac{1}{2f (q_i)}\left(\tfrac{\dot q_i}{\prod_{s=1; s\ne i}^N(q_i-q_s)}\right)^2 = -U (q_i)+ H_0 q_i^{N-1}+ \cdots + H_{N-2}  q_i  +  H_{N-1} .  
\end{equation}
We see that   the  condition \eqref{eq:condition} implies that the solutions of both 
systems viewed now as curves on $\mathbb{R}^N(q)$ satisfy the same  
 system of  ODEs  and therefore every solution of one system is a solution of the other. 
\end{proof}

Let us now compare, with the help of Lemma \ref{lem:2} and Remark \ref{rem:1}, the systems in question: the Neumann system, the geodesic flow of the metric geodesically equivalent to  the metric of ellipsoid, and the Benenti system coming from the finite-dimensional reduction of the KdV system via Theorem \ref{thm:1}.  

For all of them, the product $f(t) (-U(t) + H_0 t^{N-1}+\cdots+H_{N-1})$ is a polynomial of degree $2N+1$ with positive leading coefficient.
Clearly, the value of the leading coefficient is not important, once it is positive, as one can scale it by a positive constant by rescaling the time.  The polynomial $f(t) (-U(t) + H_0 t^{N-1}+\cdots+H_{N-1})$  has certain  restrictions which we will discuss now.

The   restriction   in the KdV systems is that    the second  highest coefficient of $f(t) (-U(t) + H_0 t^{N-1}+\cdots+H_{N-1})$   is zero. This restriction does not affect anything as the transformation 
$L_{new}= L_{new} + \const\Id $ can make   the second highest coefficient arbitrary.

Next, in the Benenti systems  describing the Neumann system and the geodesic flow of the metric geodesically equivalent to the ellipsoid, the product $f(t) (-U(t) + H_0 t^{N-1}+\cdots+H_{N-1}) $    has at least 
$N+1$ real different  zeros $a_1,\dots, a_{N+1}$.  This restriction is essential.

Note also that not all values of the integrals allow solutions. Indeed, the system of 
ODEs \eqref{eq:legendre} has no real solution if $f_\alpha(t) (-U(t)+ H_0qt^{N-1}+\cdots + H_{N-1})  <0$.

As it is clear from the introduction, for our paper the relation of the stationary solutions of KdV and the solutions of the Neumann systems is most important. Lemma \ref{lem:2} implies 
that stationary solutions of KdV, viewed as curves $x\mapsto u(x)$, correspond to the  evolution of $2\textrm{trace}(L)$ along the solution of the Neumann system. 
By Remark \ref{rem:2}, up to addition of a constant, $-\textrm{trace}(L)$ is the potential energy of the Neumann system. This relates  the potential energy of the Neumann system evaluated along a solution of the Neumann system to stationary solutions of KdV.  The relation is of course known \cite{moser80,moser81,veselov80}. 

Concerning the relation of the geodesics  on  ellipsoid to stationary solutions of the KdV system, Lemma \ref{lem:2} and the discussion above established such a relation between the metric geodesically equivalent to the metric of ellipsoid and the KdV system. The solutions of the Hamiltonian systems corresponding to the geodesic flow of the ellipsoid and to the geodesic flow of the metric geodesically equivalent to the metric of the ellipsoid give  the same curves on the manifold, but differently parameterised. And indeed, in \cite[\S 3.5]{moser80} the reparameterisation is given.

\begin{remark} \label{rem:3} 
Lemma \ref{lem:2} and discussion above links the solution  of Neumann system, geodesics of the metric geodesically equivalent to the metric of  the ellipsoid and the finite-dimensional reductions of the KdV system. 
  The arguments used in the discussion are essentially local calculations in a special coordinate chart, which, e.g.  in the case of Neumann system,   covers    almost whole but not the whole sphere. A natural question is therefore whether the links  between systems  survives  when the solution of the Neumann system leaves the coordinate chart, or possibly never passes through the coordinate chart. 

  The answer is ``yes''. Indeed,  all considered  systems are real-analytic, so if two solutions coincide locally they coincide globally. Moreover, the 
  limit of a sequence of solutions is a solution, so even if the solution never passes through the special coordinate chart, the link remains. 
\end{remark}

\section{Stationary solutions  corresponding to the Neumann system and  finite-bandness of the  corresponding Sturm-Liouville-Hill problem}

\subsection{ The property of the polynomial C(t) for solutions of KdV coming from the solution of Neumann system.  }

By Lemma \ref{lem:2}, see also discussion  short after,  the solutions of the Neumann system, viewed as curves on the $N$-dimensional  sphere, are closely related to finite-dimensional reductions of the KdV system corresponding to the polynomial 
\begin{equation}\label{eq:pol1} 
C(t)=\tfrac{1}{2} (t^N+ H_0 t^{N-1}+ \cdots + H_{N-2} t  +  H_{N-1})\prod_{s=1}^{N+1}(t-a_i).    \end{equation}
From the formula   \eqref{eq:pol1}  we immediately see that the polynomial $C(t)$ has $N+1$ real roots $a_1,\dots, a_{N+1}$.  The condition that the sphero-ellipsoidal coordinates $q_i$ used for the description of the Neumann system are necessarily real-valued and satisfy, see Remark \ref{rem:2a},  
\begin{equation}\label{eq:ax}
a_1 < q_1 <a_2 < q_2 <\cdots <q_{N}< a_{N+1},
\end{equation}
give further assumptions on the roots of $C(t)$, which we discuss now.

\begin{lemma} \label{lem:3}
   Consider a solution of the Neumann system and the corresponding polynomial  $C(t)$ given by \eqref{eq:pol1}, where $H_0, H_1, \dots, H_{N-1}$ are the values of the integrals $\tilde I_0,\dots, \tilde I_{N-1}$ corresponding to this solution. 
   
   Then, all roots of polynomial $C(t)$ are  real numbers.  
   Moreover,  if we denote roots, counted with their multiplicities, by $r_1\le r_2 \le \cdots \le r_{2N+1},  $ and by $q_1\le \cdots \le q_N$ the eigenvalues of the (1,1)-tensor $L$, again counted with their multiplicities, then 
   \begin{equation}\label{eq:ordered}
       r_1 \le  r_2  \le  q_1 \le     r_3 \le r_4\le  q_2 \le r_5 \le r_6 \le q_3   \le  \cdots \le r_{2N-2}  \le r_{2N-1} \le  q_{2N}   \le r_{2N+1},  
   \end{equation}
   that is,   $q_i$ lies on the interval $[r_{2i}, r_{2i+1}]$). 
\end{lemma}

We recall that in the sphero-ellipsoidal coordinate system $q_1,\dots,q_N$
the operator $L$ has the form $\textrm{diag}(q_1,\dots,q_N)$, so  the notation $q_i$ used for eigenvalues in Lemma \ref{lem:3} is compatible with that in \S  \ref{sec:neumann}. Note though that for the Neumann system the operator $L$ is defined also at the points where the sphero-ellisoidal coordinate system is not defined. These points are precisely the points where $L$ has multiple eigenvalues.

\begin{proof}   As the property of a polynomial to have a complex root is an open property, we may work at 
 a point $(q_1,\dots ,q_N)$ of our sphero-ellipsoidal coordinate system and assume that all component of 
 the  velocity vector $ (\dot q_1,\dots, \dot q_N)$ are different from $0$.  From \eqref{eq:pol1},  
we see that the  leading coefficient of the  polynomial  is positive  so for $t<< -1 $  we have $C(t)<0$.  Next, from  \eqref{eq:legendre} we see that the value of  $C(t)$ are negative  at each    $t= q_i$. Since 
 the polynomial $\prod_{s=1}^{N+1}(t-a_i) $ has different signs at $t=x_i$ and at $t=q_{i+1}$,   the polynomial $(t^N+ H_0 t^{N-1}+ \cdots + H_{N-2} t  +  H_{N-1}) $ should also have 
 different signs at $q_i$ and at $q_{i+1}$ implying the existence, in view of \eqref{eq:ax},  of $N-1$ real roots on the interval $[a_1,a_{N+1}]. $  Thus, $2N$  roots, counted with multiplicities,  of our polynomial
 $C$ of degree $2N+1$,   are real, so the remaining root is real as well.  
 \end{proof}

By Theorems \ref{thm:1},  \ref{thm:2}, solutions  of the KdV system  coming from finite-dimensional reduction  correspond to the solutions of the Benenti system with $f(t)=1$ and   $U(t)= -(t^{2N+1}+ c_{2}t^{2N-1}+\cdots + c_{2N+1})$.   By Lemmas    \ref{lem:2}, every solution of the Neumann system corresponds to a solution of the Benenti system above.   Lemma  \ref{lem:3} tells which solutions of the  Benenti system above correspond to the solutions of some Neumann system: necessary  conditions are that all roots of the polynomial $C(t)$ are real  and that the eigenvalues of the corresponding $L$ are real and satisfy \eqref{eq:ordered}. The next Lemma shows, in particular,  that these are also sufficient conditions.  
It also  implies that if all eigenvalues of $L$ are real and satisfy \eqref{eq:ordered} at one point of the solution of the Benenti system  above, then it is true at every point. 

\begin{lemma}\label{lem:behaviour}
    Suppose $(w_1(x),...,w_n(x))$ is a solution of \eqref{eq:base1} such that the polynomial $C(t) $ has only real roots which we list with their multiplicities and denote by $r_1\le \cdots\le  r_{2N+1}$.
Denote by $q_1(x),\dots,q_N(x)$ the zeros of the polynomial   \eqref{eq:11a}, listed with their multiplicities. Assume that  $(w_1(0),\cdots, w_N(0))$  are such that  all $q_i(0)$ are real valued and satisfy  $r_{2i} \le q_i(0)\le r_{2i+1}$. Then, the  solution $(w_1(x),...,w_n(x))$, with the initial values  $(w_1(0),\cdots, w_N(0))$, can be extended   for all $x\in \mathbb{R}.   $   

Moreover,    suppose  for a certain $i$  and for certain $\tilde x$ 
 we have $r_{2i} < q_i(\tilde x)< r_{2i+1} $ and denote by  $I_i$   the connected component of 
the set $ \{x\in \mathbb{R} \mid r_{2i}< q_i(x)<r_{2i+1}\}$ containing $0$.

Then  the    function  $q_i(x)$ smoothly depends on $x$ for $x\in I_i$, moreover    $\tfrac{d}{dx}q_i\ne 0$ for all $x\in I_i$ and satisfies \eqref{eq:legendre}.
\end{lemma}

\begin{proof} First assume that for each $i$  we have  $r_{2i}<q_i(0)<r_{2i+1}$.  Then, we take the Neumann system such that $a_i=r_{2i}$ for $i<N+1$ and  $a_N= r_{2N+1}$. Next, take a point on $S^N$ such that sphero-ellipsoidal coordinates of this points are $(q_1(0),\dots, q_{N}(0)).$  The corresponding function $f(t)$ in the equation \eqref{eq:legendre} is then equal to $-4 \prod_{s=1}^{N+1} (q_i-a_s)$, see \eqref{eq:fsphere}. Next, take $H_0,\dots, H_{N-1}$ such that the right hand side of the equation \eqref{eq:legendre} has roots $r_1,r_3,\dots, r_{2i-1},\dots,  r_{2N-1}$. This is clearly possible since for the Neumann system $U(t)$ is given by \eqref{eq:Usphere}, so 
    the right hand side of \eqref{eq:legendre} is polynomial of degree $N$ in which we freely choose all nonleading coefficients.     Then, at out point $(q_1(0),\dots, q_{N}(0))$ of the sphere the products 
    $f(q_i(0))U(q_i(0))$ are positive and we can choose $\left(\tfrac{\ddd}{\ddd x} q_1(0),\dots, \tfrac{d}{dt} q_N(0)\right)$ 
    such that the equation \eqref{eq:legendre} is satisfied. The solution of the Neumann system with this initial data will correspond to the solution of \eqref{eq:base1}
    with the initial values  $(w_1(0),\cdots, w_N(0))$. Since the solution of Neumann system is defined on a closed manifold and therefore can be extended for the whole $\mathbb{R}$, the solution of   \eqref{eq:base1} can be extended to the whole $\mathbb{R}$ as well.  At the points of the interval $I_i$, the  function  $C(q_i(x))$ is not zero, so \eqref{eq:base1} implies that $\tfrac{d}{dt} q_i(x)\ne 0$.   Lemma \ref{lem:behaviour} is proved  under the assumption that all roots of the polynomial  $C$ are different. 

    The remaining  case when for a certain $i$ we have  $r_{2i}=q_i(0)$ or $r_{2i+1}=q_i(0)$  can  be proved by passage to the limit. We take  the sequence of the polynomials $C_n(t)$ satisfying the standard assumptions  such  that it converge to our polynomial $C(t)$, and the sequence of initial data satisfying assumptions used in the proof above  such that it converges to our initial data.   The limit of the corresponding solutions   of the Neumann system is a solution of the Neumann system. As it lives on a compact manifold, it exists for all  
    $x\in \mathbb{R}$.

    Finally, for $x\in I_i$,   $q_i(x) $ is an isolated root so it depends smoothly on $x$.  Since the equation \eqref{eq:legendre} are fulfilled for the sequence of the solutions considered above, it is fulfilled for the limit solution. But then \eqref{eq:legendre} implies $\tfrac{d}{dx}q_i\ne 0$. 
\end{proof}

\begin{corollary} \label{cor:1}
  Any solution   of the KdV system corresponding to a solution of the Neumann system is bounded.    
\end{corollary}
\begin{proof}
 Indeed, the solution is (up to a constant) the trace of $L$, which is the sum of $x_i$, which is clearly bounded from above by $N a_{N}$ and from  below by $N a_{1}.$\end{proof} 
 
\begin{proof}[An alternative proof:] up to a constant,  the solution is essentially the potential energy of the Neumann system evaluated along a solution of the Neumann system. The potential energy is clearly bounded as it 
 is a continuous function on  the sphere  $S^N$  which is compact.  \end{proof}

\subsection{Proof that the Sturm-Liouville-Hill problem corresponding to the Neumann system has finite-band spectrum.  }
In the previous section we discussed the conditions on the polynomial $C(t)$  for the stationary solutions of the KdV equations coming from the Neumann system, in particular, they should have only real roots. 
The goal of this section  is to prove the  following Theorem: 

\begin{theorem}\label{lem:p3}
Denote by $r_1\le r_2\le \cdots \le r_{2N+1} $ the roots of the polynomial $C$, counted with their multiplicities. 

Then, $\tfrac{1}{2} \lambda$ lies in  the spectrum of $ -\tfrac{\partial }{\partial x^2} + \tfrac{1}{2} u$ if and only if   
\begin{equation} \label{eq:bands} \lambda \in [r_1, r_2] \cup [r_3, r_4]\cup \cdots \cup [r_{2N-1}, r_{2N}] \cup [r_{2N+1}, +\infty). \end{equation} 
\end{theorem}

We see that the spectrum consists  of finitely many closed intervals and one half-line $[2r_{2N+1}, +\infty)$.    In jargon, the intervals are called {\it bands} and the  half-line $[r_{2N+1}, +\infty)$ is called {\it the  infinite band}.  
Note that the number of (finite)  bands  is at most $N$.  Since $r_i$ may be equal to $r_{i+1}$, the number of bands can be smaller than $N$ and some of the bands could be actually points.  

Let us also observe that the set \eqref{eq:bands} coincides with 
\begin{equation} \label{eq:bands1} \{\lambda \in \mathbb{R} \mid  C(\lambda)\ge 0\}\, .  \end{equation}

 Lemma \ref{lem:p2}  and the discussion afterwards relates the KdV equation, via \eqref{eq:base1}, to the spectral problem for Schr\"odinger-Hill equation.

\begin{lemma}\label{lem:p2}
Let    $w(x ), \sigma(x )$  be   smooth functions defined on an interval such that $w>0$  
and $m\ne 0$, $C$ be constants.   
 Consider the primitive function $\phi(x) = \int \frac{\ddd s}{2 w(s)}$ and    the functions
\begin{equation}\label{eq:sub}
\psi_{1} (x) =  \sqrt{w(x)} \, e^{\sqrt{- \frac{2C}{m}} \phi(x)} \quad \text{and} \quad \textrm{sign}(w(x)) \,\psi_{2} (x) = \sqrt{w(x)} \, e^{- \sqrt{- \frac{2C }{m}} \phi(x)}.
\end{equation}
Then, equation
\begin{equation}\label{eq:lin}
 \psi_{i, xx} + \frac{1}{2} \frac{\sigma(x)}{m} \psi_i = 0, \quad i = 1, 2   
\end{equation}
holds if ans only if equation
\begin{equation}\label{eq:base}
 w_{xx} (x) w (x) - \frac{1}{2} w^2_x (x) + \frac{\sigma(x) w^2(x) - C }{m} = 0.  
\end{equation}
holds. 
\end{lemma}

\begin{remark} { \phantom{1}}

\begin{enumerate} \item     The  function $\phi(x)$ is     defined up to an addition of a constant, which corresponds to multiplication of the function   $\psi_1$  by a constant and division of $\psi_2$ by the same constant. 
\item   
The change of the sign of the function $w$ does not affect the   equation \eqref{eq:base}. If $w<0$ on the interval we work, the  functions  
$$ \sqrt{-w(x)} \, e^{-\sqrt{- \frac{2C}{m}} \phi(x)} \quad \text{and} \quad -  \sqrt{-w(x)} \, e^{ \sqrt{- \frac{2C }{m}} \phi(x)}$$
are solution of \eqref{eq:lin}. 
 
\item The functions $\psi_1, \psi_2 $ clearly have the property $\psi_1\psi_2= w$. One easily shows that,  within the 2-dimensional solution space of  \eqref{eq:lin}, the property $\psi_1\psi_2= w$ defines (up to swapping, multiplication one by a constant and division of the second by the same constant) two solutions of \eqref{eq:lin} which form a basis in the space of all solutions. 
\end{enumerate}
\end{remark}

\begin{proof}[Proof of  Lemma \ref{lem:p2}]  Note that the equations   \eqref{eq:lin} and \eqref{eq:base} are both second order ODEs so substituting  $\psi_i$ given by \eqref{eq:sub} in \eqref{eq:lin} will give a second order differential equation on $w$. We will show that the formula for $\psi_i$ is chosen such that the obtained equation is equivalent to \eqref{eq:base}. 

Let us do the corresponding calculations.       
To simplify the formulas, we assume $w>0$, denote $k = \sqrt{- 2\frac{C }{m }}$  and consider     $\psi:= \psi_1$. We get
$$
\begin{aligned}
\psi_x & = \frac{1}{2} \frac{w_x}{\sqrt{w}} e^{k\phi} + \frac{k}{2\sqrt{w}} e^{k\phi}, \\
\psi_{xx} & = \frac{1}{2} \frac{w_{xx}}{\sqrt{w}} e^{k\phi} - \frac{1}{4}\frac{w_x^2}{w^{3/2}} e^{k\phi} + \frac{k}{2} \frac{w_x}{w^{3/2}} e^{k\phi} -\frac{k}{2} \frac{w_x}{w^{3/2}} e^{k\phi} + \frac{k^2}{4 w^{3/2}}e^{k\phi} - \frac{1}{2} \sigma(\lambda) \sqrt{w} e^{k\phi} = \\
& = \frac{1}{2}\Bigg( \frac{w_{xx}}{w} - \frac{1}{2} \frac{w_x^2}{w^2} - \frac{C}{m w^2}\Bigg) \sqrt{w} e^{k\phi} = - \frac{1}{2} \frac{\sigma}{m} \psi_1.
\end{aligned}
$$
Thus, the equation \eqref{eq:lin} holds for the function $\psi_1$. The proof for   $\ \psi_2$ is analogous.

Let us now prove the Lemma in the other direction. Assume that \eqref{eq:lin} holds. The  same computation give us 
$$
0 = \Bigg( \frac{w_{xx}}{w} - \frac{1}{2} \frac{w_x^2}{w^2} - \frac{C}{m w^2} + \frac{\sigma}{m} \Bigg) \psi_1. 
$$
As $w$ is not zero by assumptions, \eqref{eq:sub} implies that $\psi_1$ is not zero so 
 \eqref{eq:base} holds.  
\end{proof}

The equation \eqref{eq:lin} is just the Sturm-Liouville equation corresponding to  the potential   $\frac{1}{2} \frac{\sigma(x )}{m}$; its operator is 
$$
\mathcal{H} = \tfrac{\partial^2 }{\partial x^2} + \frac{1}{2} \frac{\sigma(x )}{m } \operatorname{Id}.
$$
In this  notation,  the equation \eqref{eq:lin} reads 
\begin{equation}\label{eq:lin2}
\mathcal{H }\psi = 0.    
\end{equation}

Next, note that \eqref{eq:base} 
is just the equation  \eqref{eq:base1} with  $m = m_0$ and $\sigma(x,\lambda)= \lambda -\tfrac{1}{2}w_1 .$ 
We assume $m_0=1$ for simplicity, see the discussion around \eqref{eq:kdv}.  In this case, \eqref{eq:base} has the   form
\begin{equation}\label{k:base2}
w_{xx} (x, \lambda) w (x, \lambda) - \frac{1}{2} w^2_x (x, \lambda) + \underbrace{\sigma(x, \lambda)}_{ {\lambda}- \tfrac{1}{2} w_1 } w^2(x, \lambda) = C(\lambda).   
\end{equation}
(The importance of the last equation was first emphasized in \cite{GD75}).
In view of Theorem \ref{thm:1}, the corresponding Sturm-Liouville equation from  Lemma \ref{lem:p2} reads then 
\begin{equation} \psi_{xx}+ \tfrac{1}{2}\left( {\lambda}- u \right) \psi    =0. \label{eq:SL} \end{equation}
It is the so-called Hill-Schr\"odinger equation, i.e., the equation on the eigenfunctions with   eigenvalue $ \tfrac{1}{2} \lambda   $
for the  operator   
\begin{equation}
  \label{eq:hs}
\left(-\tfrac{\partial }{\partial x^2} + \tfrac{1}{2} u\right)\psi =  \tfrac{1}{2} \lambda \psi. 
\end{equation}

Note that \eqref{eq:hs} is an ordinary linear  differential equation of the  second order, so it has  precisely  two linearly independent solutions for each $\lambda$, on any connected open  interval where  $u$ is defined.  
We say that  $ \tfrac{1}{2} \lambda\in \mathbb{R}$ lies in {\it spectrum}, if $u$ is defined on the whole $\mathbb{R}$ and there exists a \underline{bounded} non-zero solution $\psi$ of \eqref{eq:hs}.

\begin{proof}[Proof of Theorem \ref{lem:p3}]
First, let us show, that the bands, indeed, lie in the spectrum.   First, observe that the roots  $r_1,...,r_{2N+1}$ lie in the spectrum. Under the assumption $\lambda=r_i$, both solutions from  \eqref{eq:sub} coincide and give us  a solution 
\begin{equation}\label{eq:sub1}
\psi  (x, \lambda) = \sqrt{|w(x, \lambda)|} 
\end{equation}
which is bounded in view of Corollary \ref{cor:1}.

Next, consider  $\lambda$ from the  open interval    $(r_{2i-1}, r_{2i}) $.   For  each  $\lambda \in (r_{2i-1}, r_{2i})  $ we have $C(\lambda) >  0$, as no root of the characteristic 
polynomial of $L$ lies in such an interval. Therefore,  $w(x, \lambda)$ does not change the sign and we may   assume $w(\lambda) >0$.  Since $C(\lambda)$ is positive, $\sqrt{-2C(\lambda)}$ is purely  imaginary so  
$e^{\sqrt{-2C(\lambda)} \phi(x) }$ is bounded. Then,  the solution  $\psi_1=\sqrt{w(x)} e^{\sqrt{-2C(\lambda)} \phi(x) }$  from \eqref{eq:sub} is bounded as well. Thus, $\tfrac{1}{2} \lambda $ lies in the  spectrum. 

The infinite band can be handled using a similar argument: Take 
$\lambda \in [r_{2N+1}, +\infty)$. We  again have that   $\lambda$ cannot be a root of   the characteristic 
polynomial of $L$. Therefore,  $w(x; \lambda)$ does not change the sign and we may think that it is positive.  Since $C(\lambda)$ is positive, $\sqrt{-2C(\lambda)}$ is purely  imaginary so  
$e^{\sqrt{-2C(\lambda)} \phi(x) }$ is bounded. Finally,  the solution  $\psi_1=\sqrt{|w(x)|} e^{\sqrt{-2C(\lambda)} \phi(x)} $  from \eqref{eq:sub} is bounded  as well. Thus, $\tfrac{1}{2} \lambda $ lies in the  spectrum.

Next, let us show that  for  $\lambda \in \mathbb{R}$ such that $C(\lambda)<0$   
the ODE  \eqref{eq:hs} does  not have bounded solutions.  We start with the case when  $\lambda<r_1$. In this case, $w(x)$    is separated  from zero, we may assume that it is positive.
 The function $\phi(x)= \int  \tfrac{\ddd s}{2w(s)}  $ is therefore unboundend.  Since $\sqrt{-2C(\lambda)}$ is a real positive number, the function $\psi_1$ from \eqref{eq:sub} is unbounded for $x\to \infty $ and goes to zero for $x\to -\infty$. The function  $\psi_1$ from \eqref{eq:sub} is unbounded for $x\to -\infty $ and goes to zero for $x\to  \infty$. Therefore, no  nontrivial linear combination  of $\psi_1$ and $\psi_2$  is bounded.

Finally, let us consider the  most complicated case   $\lambda\in (r_{2i}, r_{2i+1})$. Without loss of generality, we may assume that $i=1 $    and $\lambda=0$, so $r_2<0$ and $r_3>0$.
As above, we  denote by $q_1,\dots, q_N$ the
eigenvalues of $L$ and consider their dependence on the ``time'' $x$, that is, the functions 
$q_i(x)$. We    reserve the notation $\tilde x $ for values such that 
  $q_1(\tilde x)=\lambda=0$.  Note also that the assumption $\lambda=0$  implies $w(x)= q_1(x)q_2(x)\cdots q_N(x)$ with all  $q_i\ge  r_3>0 $ for $i \ge 2$. The functions $q_i$ are  well-defined at least continuous functions. 
We know from Lemma \ref{lem:behaviour} that the derivative  $\tfrac{\ddd q_1}{\ddd x}$ is not  zero  at $ \tilde x$ such that 
$q_1(\tilde x)=0$. Actually, by \eqref{eq:legendre},  it is given 
at any point $x$   by 
\begin{equation}\label{eq:dq1}
    \tfrac{\ddd q_1(  x)}{\ddd x}  =\pm \frac{\sqrt{-2C(0) }}{ \left(q_2(x)-  q_1(x)\right)\left(q_3(x)-q_1(x)\right) \cdots \left(q_N(x)-q_1(x)\right)}, 
\end{equation} 
so at $x=\tilde x$
it is given by 
\begin{equation}\label{eq:dq}
    \frac{\ddd q_1(\tilde x)}{\ddd x}  =\pm \frac{\sqrt{-2C(0) }}{q_2(\tilde x)q_3(\tilde x)\cdots q_N(\tilde x)}. 
\end{equation} 
This in particular implies that, at the point $\tilde x$, the derivative of $w$ is given by 
\begin{equation}\label{eq:dw}
    \frac{\ddd w(\tilde x)}{\ddd x}  =  \frac{\ddd q_1(\tilde x)q_2(\tilde x)\cdots q_N(\tilde x)}{\ddd x} = q_2(\tilde x)\cdots q_N(\tilde x)\frac{\ddd q_1(\tilde x)}{\ddd x} = 
    \pm \sqrt{-2C(0) } . 
\end{equation}

We see that as $C(0)\ne 0$,  the derivative of $q_1$ and of $ w$ is not zero at $\tilde x$. The sign 
 $\pm$  depends on whether  the functions $q_1, w$ are locally increasing or decreasing at a small neighborhood of $\tilde x$. Note also that the derivative of $w$ at the points such that $q_1=0$ has the same absolute value    and its sign depends on whether the sign of $w$ changes from ``$+$'' to ``$-$'' or from  ``$-$'' to ``$+$'' when $x$ passes this point. 

This in particular implies that the points $\tilde x$ such that $q_1(\tilde x)=0$ exist, as the derivative of the function
$q_1(x)$ is nonzero for $x$ such that $r_1<q_1(x)<r_2\, , \ q_1(x)\ne 0$ by  \eqref{eq:legendre}.  

Lemma \ref{lem:p2} gives us two eigenfunctions $\psi_1$ and $\psi_2$ 
of the equation \eqref{eq:hs} at the interval  where $w$ does not change sign. 
Note that as $\psi_1, \psi_2$ are solutions of a linear ODE with smooth coefficients, they can be extended to solutions  defined and smooth on the whole $\mathbb{R}$.  The solutions $\psi_1, \psi_2$ are linearly independend and therefore form a basis in the two-dimensional linear space of the solutions of \eqref{eq:hs}.

 Let us understand how the solutions $\psi_1,\psi_2$ behave near zeros of $w$. At a zero $\tilde x$ of $w$, using \eqref{eq:dq},
 we  have that the 
  primitive function $\phi(x)= \int^x \sqrt{-2C(0)} \tfrac{1}{2 w(s)} ds $ behaves  asymptotically, for $x\to \tilde x$, as $\textrm{const}  \pm   \tfrac{1}{2}\ln |w(x) |$.   This implies that one of the functions $\psi_1(x)$ and $\psi_2(x)$  behaves asymptotically, for $x\to \tilde x$,  as  
$$
  \textrm{const} \cdot w    \ \ \textrm{and  another as }  \  \  \  1/\textrm{const  }   
$$
(for a certain nonzero constant $\textrm{const}$ which of course may depend on the choice of the point $\tilde x$ at which $q_1=0$).

Next, observe that as objects we considered were real-analytic, we can think that our parameter $x$ is complex-valued, $x\in \mathbb{C}$. The  solutions $w(x)$ and $q_1(x)$ are then holomorphic functions of $x\in \mathbb{C}$. Since the derivatives of $q_1$ and of $w$ are not zero at the points $\tilde x$ such that $q_1=0$, such points are isolated in a small neighborhood of the real line.  We denote this neighborhood by $W$.

We denote by $\dots <\tilde x_{-1}<\tilde x_0<\tilde x_1<\tilde x_2<\dots $ the points of $\mathbb{R}$ 
such that $q_1(\tilde x_i)=0$, the sequence of such $\tilde x_i$ could be unbounded in both directions, in one direction, or can simply be finite.  
Assume without loss of generality 
that $w>0$ on the  interval $(\tilde x_0, \tilde x_1)\subset\mathbb{R}$. Then it changes sign from ``$+$'' to ``$-$'' at   the points $x_0, x_1$, so it is negative on the intervals $(x_{-1}, x_0)$ and $(x_1, x_2)$,   positive on the intervals 
$(\tilde x_2, \tilde x_3)$,  $(\tilde x_{-2},\tilde x_{-1})$,  negative on   the intervals
$(\tilde x_{-3}, \tilde x_{-2})$,  $(\tilde x_{3},\tilde x_{4})$ and so on. The  intervals such that $w$ is positive will be called {\it positive intervals, } and such that  $w$ is negative  will be called {\it negative intervals. }

Next, denote by $\widetilde W$ the neighborhood $W$ without negative intervals. As at the endpoints of the negative intervals the the derivatives of $w$ at $\tilde x_i$ and $\tilde x_{i+1}$ coincide in absolute values but have different signs by \eqref{eq:dw}, the  corresponding  residues cancels. Then, the primitive function $\phi= \int \tfrac{1}{2w(s)}\ddd s$ from Lemma \ref{lem:p2} is well-defined on the whole $\widetilde W$. This implies that the formula \eqref{eq:sub} globally   defined the restriction of global solutions 
   of \eqref{eq:lin}
   to $\widetilde W$.

  Recall now that the product $\psi_1 \psi_2$ equals $\omega,$ and at the points $x\in \mathbb{R}$ 
  with $|q_1(r)|>\varepsilon>0$  the function $w$ is separated from zero.
  Our next goal is to  show now that for $x\in \mathbb{R}_{>0}$ the function $\psi_1$ achieves arbitrary big  values.
  Recall that the equation \eqref{eq:lin} is a linear second order ODE and its solution space is 2-dimensional, so any   solution  $\psi$   on $\widetilde W$ is given by  \begin{equation} \label{eq:psis} \const_1 \psi_1 + \const_2 \psi_2 .\end{equation}
  At the points where $\psi_1$ is big in its  absolute value,  the function $\psi_2$ is small so a bounded linear combination $\const_1 \psi_1 + \const_2 \psi_2 $   for unbounded $\psi_1$ implies $\const_1=0$.

In order to do it,  let us  decompose the function   $w(x) $ in the product $ q_1(x) \bar w(x)$, where $\bar w(x)= q_2q_3\cdots q_N$. In our setup, $\bar w(x)$  is positive for all $x$.   
Consider  a pair    $(a(x),b)$, where $a$ is a function and $b$ is a  sufficiently small  positive constant   
such that \begin{equation}\label{eq:con1} 
a(x) \bar w(x) + b q_1(x) = 1.\end{equation}
The existence of such a pair $(a(x), b)$ is clear. Indeed, if we take $b>0$  sufficiently small, then $1 - b q_1(x) > \varepsilon >0$ for all $x$. Then the function $a(x)$ is given  by
$$
a(x) = \frac{1 - b q_1(x)}{\bar w(x) }. 
$$ 
If $b$ is small enough,  $a(x)$ is   separated  from zero. 

The condition \eqref{eq:con1} implies that
$$
\frac{1}{q_1(x) \bar w(x)} = \frac{a(x)}{q_1(x)} + \frac{b}{\bar w(x)}.
$$
Combining this formula with \eqref{eq:sub},  we obtain 
\begin{equation}\label{eq:rb}
\psi_1 = \left(\sqrt{q_1} e^{\sqrt{- 2 C(0)} \int \frac{a(s)}{2 q_1(s)} \ddd  s } \right) \cdot \left(  {\sqrt{\bar w} e^{\sqrt{- 2 C(0)} \int \frac{b}{2 \bar w(s)}\ddd s}}\right).
\end{equation}
The second factor on the  right hand side of \eqref{eq:rb} is unbounded for $x \to +\infty$ as the function  $\frac{b}{2 q_1(s)}$ is positive,  $\sqrt{-2C(0)}$  is positive and $\bar w$ is separated from zero. Note that the second factor $$\left(  {\sqrt{\bar w} e^{\sqrt{- 2 C(0)} \int \frac{b}{2 \bar w(s)}\ddd s}}\right)$$ can be extended to the whole real line $\mathbb{R}, $ as the integrand  $ \tfrac{b}{2 \bar w(s)}$  has no poles. 

The function $q_1(s)$ in the first factor must, for certain arbitrary large  $s$, be different from  zero, as its derivative is different from zero   by   Lemma \ref{lem:behaviour}
at the points where $q_1=0$.  Its derivative $\tfrac{\ddd}{\ddd x}q_1(x)$  satisfies \eqref{eq:legendre} and therefore for every fixed  $\tilde q_1\in (r_1, r_2)$   the derivative  $\tfrac{\ddd}{\ddd x}q_1(x)$  is bounded from zero on the set   $\{s\in \mathbb{R}  \mid q_1(s)=\tilde q_1\}$. Then,   the first bracket   \eqref{eq:rb}  cannot be arbitrarily close to zero implying that  \eqref{eq:rb} unbounded for $x\to +\infty$.  As explained above, this means that $\const_1$ in \eqref{eq:psis} is zero, 

Similarly,   the function $\psi_2$ is bounded for $x\to -\infty$ implying $\const_2=0$.  Theorem is proved.
\end{proof}

\section{Conclusion and outlook} \label{sec:con}

The goal of this note was  to understand the famous results of \cite{novikov74, moser80, moser81} using new approach and new methods coming from  recent investigation \cite{nijapp4,BKMreduction} of   
BKM systems.   We have shown that it is possible, at least for a part of results, and  the proofs in the present paper are  shorter   than that of \cite{moser80,moser81}. Of course the proofs use the preliminary work done in \cite{nijapp4,BKMreduction}. But still, at least in our eyes, we replaced the ``mathematics magic''  by ``mathematical methods''; in particular the wonderful observations of H. Kn\"orrer \cite{knorrer82} relating  the geodesics of ellipsoid to the solutions of the Neumann systems, and of Veselov \cite{veselov80} relating finite-gap solutions of the KDV to the solutions of the Neumann system are direct corollaries of Lemma \ref{lem:2}, which can be applied to a much wider class of systems. 

But in fact the main motivation behind this investigation  is as follows: the de-facto  main approach of  tackling  KdV equations is  the so-called  inverse scattering method. Inverse scattering method
is based on the understanding  of the 
eigenfunctions of the Schr\"odinger operator  $-\tfrac{\ddd^2}{\ddd x^2}  +  u(t, x)$, where $u(t, x)$ is essentially  a solution of the KdV equation, and studying their evolution  when we change $t$.  It is clearly related to the problem we consider, and actually the results  and methods of \cite{novikov74, moser80, moser81} came from   this approach, see also \cite{BM2006,MB2010, BM2008} for recent developments. 
A natural question on which we will concentrate in our next investigation  is how to  adapt 
the inverse scattering method for studying   general BKM  systems, with the goal  to  generalise   famous results on the KdV equation for all BKM systems.   

Let us indicate  first  observation  in this direction. 
First let us recall that in the finite-dimensional reduction of BKM systems studied in \cite{BKMreduction}  the analog of \eqref{eq:base1}  is   
\begin{equation}\label{eq:base_new}
  w_{xx} (x, \lambda) w (x, \lambda) - \frac{1}{2} w^2_x (x, \lambda) + \frac{\sigma(w, \lambda) w^2(x, \lambda) - C(\lambda)}{m(\lambda)} = 0.   
\end{equation}
Above $m(\lambda)$ is a polynomial of degree $\le n$ with constant coefficients used for the construction of the BKM system, $C$ is a polynomial of degree $2N+1$, $w=w(x)$ is the  monic  polynomial \eqref{eq:11a} of degree $N$. The function $\sigma(w,\lambda) $ is  the monic polynomial of degree $n$ in $\lambda$, whose coefficients are uniquely defined by the condition that $\frac{\sigma(w, \lambda) w^2(x, \lambda) - C(\lambda)}{m(\lambda)}$ is a monic polynomial of degree $2N-1$. They are functions of $w_1, w_2,\dots, w_n$. 

The system \eqref{eq:base_new} is polynomial in $\lambda $ of degree $2N-1$, so it is equivalent to a system of $2 N $ ODEs of the second order  on $N$ functions $w_1,\dots, w_N$. This system  is equivalent to the system of Euler-Lagrange equation, whose Lagrangian is the sum of the kinetic energy coming from the flat metric \eqref{eq:comp} and  a certain potential energy depending on parameters.
   As in the KdV case, the system is a Benenti system. The potential is slightly more complicated compared to that  in the KdV case and the corresponding function $U$ is  a  polynomial of a  possibly higher degree,   if $m(\lambda)=\const$, 
   or a rational function if $m(\lambda)\ne \const$.

Having a solution of this system of ODEs, one constructs a solution of the corresponding  BKM system  as follows: recall that  in order to choose a BKM system  we should choose  
a differentially nondegenerate Nijenhuis operator  $\tilde L$ 
on $\mathbb{R}^n $ with coordinates $u_1,\dots,  u_n$. The solution $(u_1,\dots, u_n)$
corresponding to the BKM system is constructed by the solution $w$ of the  \eqref{eq:base_new} by the formula 
\begin{equation} \label{eq:sigma} 
\sigma(w, \lambda)= \det(\lambda \Id - \tilde L).  \end{equation}
 For a given $\tilde L$, the  right hand side of \eqref{eq:sigma}  is an explicit expression in $u$ from which we may reconstruct $u$.  In particular, if in the coordinates $u$ the operator $\tilde L$  has the companion form \eqref{eq:comp} (with $w$ replaced by $u$ and $N$ replaced by $n$), then  $\det(\lambda \Id - \tilde L)= \lambda^n + u_1\lambda^{n-1}+\cdots+ u_n$, so the formula   \eqref{eq:sigma}  gives us an explicit formula for $u_1,\dots, u_n$,

Note that the equation \eqref{eq:base_new} has the form \eqref{eq:base1}. The dependence on $\lambda$ is though different from that of in the KdV case. The analog of the 
equation \eqref{eq:SL} in this case is the equation 
\begin{equation} \psi_{xx}+ \tfrac{1}{2}\frac{\det(\lambda \Id - \tilde L(u(x)) }{m(\lambda)}\psi    =0. \label{eq:SL2} \end{equation}

We see that for $n\ge 2$ the dependence on the $\lambda$ is not linear, differently from  the Hill-Schr\"odinger equation.  Indeed, the function  $\tfrac{1}{2}\frac{\det(\lambda \Id - \tilde L(u(x)) }{m(\lambda)} $ is a  rational function in $\lambda$ whose coefficients  depend on  $x$. 

The special case $m(\lambda) = \textrm{const}=m_0$ of this equation   appeared in the literature, see  e.g. \cite{ af,alonso}, under the name ``energy-dependent potentials'', in the relation to multicomponent integrable systems. The scattering problem for such systems were 
studied  quite recently, see e.g.  \cite{nab,hm}.

We say that $\lambda$ lies in the spectrum of \eqref{eq:SL2}, if there exists a bounded nonzero solution $\psi$. For certain $u$ coming from finite-dimensional reductions of the BKM systems, the spectrum is the 
finite-band one, in the sense it  contains finitely many connected components.  It will be interesting to understand whether in the class of quasi-periodic coefficients of  the equation \eqref{eq:SL2}  the finite-bandness of the spectrum is equivalent to the property that the coefficients came from a BKM system. Moreover, as mentioned above, it will be generally interesting to develop the inverse scattering approach to BKM systems; the corresponding Sturm-Liouville equation is expected to be \eqref{eq:SL2}.

\subsection*{Acknowledgement. }  A.\,K. was supported by the Ministry of Science and Higher Education of the Republic of Kazakhstan (grant No. AP23483476). V.\,M. thanks the DFG (projects 455806247 and 529233771), and the ARC  (Discovery Programme DP210100951) for their support. A part of the work was done during the preworkshop and workshop on Nijenhuis Geometry and Integrable Systems at La Trobe University and the Matrix Institute. The participation of A.K. and V.M. at the workshop was supported by the Simons Foundation, and the participation of A.K. at the preworkshop was partially supported by the ARC Discovery Programme DP210100951.   We thank A. Bolsinov, H. Dullin and K. Sch\"obel for useful discussions. 
The circle of questions which lead to this note, see \cite[\S 2.4]{problist},  was motivated by discussions of V.M. with H. Dullin at the Sydney Mathematics  Research Institute.  V. M.  thanks the Sydney Mathematics  Research Institute for hospitality and for partial financial support during his visits in  2023 and 2024.   We  are very grateful to  A. Bolsinov for his careful reading of 
the almost final version of the manuscript  and finding  misprints, and for the anonymous referee for his suggestions.

\bibliographystyle{spmpsci}
\bibliography{nijenhuis}

\end{document}